\documentclass[11pt]{article}

\usepackage{wrapfig}
\usepackage{amsmath, amssymb, algorithm,thmtools,enumerate,caption,framed}
\usepackage{fullpage,amsthm}
\usepackage{paralist, sidecap}

\usepackage[usenames,dvipsnames,svgnames,table]{xcolor}
\definecolor{darkgreen}{rgb}{0.0,0,0.9}

\usepackage[colorlinks=true,pdfpagemode=UseNone,citecolor=Blue,linkcolor=Red,urlcolor=Black,pagebackref]{hyperref}

\usepackage{tikz}
\usepackage[disable]{todonotes}

\usepackage[T1]{fontenc}
\usepackage[utf8]{inputenc}
\usepackage{authblk}

\newcommand{\urg}{\texttt{NE}}
\newcommand{\ulg}{\texttt{NW}}
\newcommand{\drg}{\texttt{SE}}
\newcommand{\dlg}{\texttt{SW}}

\makeatletter
\newcommand{\setword}[2]{
  \phantomsection
  #1\def\@currentlabel{\unexpanded{#1}}\label{#2}
}
\makeatother

\numberwithin{equation}{section}

\newtheorem{theorem}{Theorem}[section]
\newtheorem{lemma}{Lemma}[section]
\newtheorem{observation}{Observation}[section]
\newtheorem{claim}[theorem]{Claim}

\title{On Guarding Orthogonal Polygons with Bounded Treewidth}

\author{Therese Biedl and Saeed Mehrabi\\

{\small Cheriton School of Computer Science,

University of Waterloo, Waterloo, Canada.

{\tt \{biedl,smehrabi\}@uwaterloo.ca}}
}

\date{}

\begin{document}

\maketitle

\begin{abstract}
There exist many variants of guarding an orthogonal polygon in an orthogonal fashion: sometimes a guard can see an entire rectangle, or along a staircase, or along a orthogonal path with at most $k$ bends.  In this paper, we study all these guarding models in the special case of orthogonal polygons that have bounded treewidth in some sense.  Exploiting algorithms for graphs of bounded treewidth, we show that the problem of finding the minimum number of guards in these models becomes linear-time solvable in polygons of bounded treewidth.
\end{abstract}

\section{Introduction}
\label{sec:introduction}
In this paper, we study orthogonal variants of the well-known art gallery problem. In the standard art gallery problem, we are given a polygon $P$ and we want to guard $P$ with the minimum number of point guards, where a guard $g$ sees a point $p$ if the line segment $\overline{gp}$ lies entirely inside $P$. This problem was introduced by Klee in 1973~\cite{orourke1987} and has received much attention since.  $\lfloor n/3 \rfloor$ guards are always sufficient and sometimes necessary \cite{chvatal1975}, minimizing the number of guards is \textsc{NP}-hard on arbitrary polygons~\cite{lee1986}, orthogonal polygons~\cite{dietmar1995}, and even on simple monotone polygons~\cite{krohn2013}. The problem is \textsc{APX}-hard on simple polygons~\cite{eidenbenz2001} and several approximation algorithms have been developed \cite{ghosh2010,krohn2013}.

Since the problem is hard, attention has focused on restricting the type of guards, their visibility or the shape of the polygon. In this paper, we consider several models of ``orthogonal visibility'',  and study orthogonal polygons that have bounded treewidth in some sense. Treewidth (defined in Section~\ref{sec:treewidth}) is normally a parameter of a graph, but we can define it for a polygon $P$ as follows. Obtain the {\em standard pixelation} of $P$ by extending a horizontal and a vertical ray inward at every reflex vertex until it hits the boundary of $P$ (see also Figure~\ref{fig:pixelAndRefinement}). We can interpret this subdivision into rectangles as a planar straight-line graph by placing a vertex at any place incident to at least two segments, and define the treewidth of a polygon $P$ to be the treewidth of the graph of the standard pixelation.

\begin{figure}[t]
\centering
\includegraphics[width=0.50\linewidth,trim=0 0 250 0,clip]{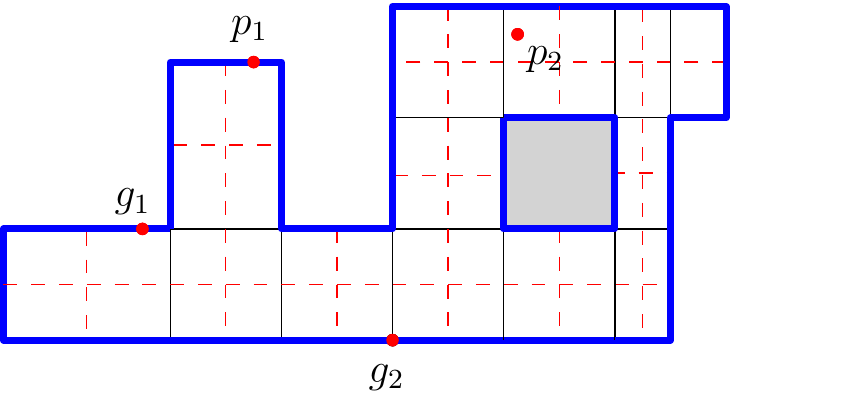}
\caption{A polygon $P$ with its standard pixelation (black, solid) and its 1-refinement (red, dashed). The gray area indicates a hole.}
\label{fig:pixelAndRefinement}
\end{figure}

\paragraph{Motivation.}
One previously studied special case of the art gallery problem concerns {\em thin} polygons, defined to be orthogonal polygons for which every vertex of the standard pixelation lies on the boundary of the polygon. Thus a polygon is simple and thin if and only if the standard pixelation is an outer-planar graph. Tom{\'{a}}s~\cite{tomas2013} showed that the (non-orthogonal) art gallery problem is \textsc{NP}-hard even for simple thin polygons if guards must be at vertices of the polygon. Naturally, one wonders whether this NP-hardness can be transferred to orthogonal guarding models. This is not true, for example $r$-guarding (defined below) is polynomial on polygons whose standard pixelation is outer-planar, because it is polynomial on {\em any} simple polygon \cite{worman2007}. But, what can be said about polygons that are ``close'' to being thin? Since outer-planar graphs have treewidth 2, this motivates the question of polygons where the standard pixelation has bounded treewidth.

The goal of this paper is to solve orthogonal guarding problem for polygons of bounded treewidth.    There are many variants of what ``orthogonal guarding'' might mean; we list below the ones considered in this paper:
\begin{itemize}
\vspace*{-2mm}
\itemsep -2pt
\item
\emph{Rectangular-guarding ($r$-guarding).} A point guard $g$ {\em $r$-guards} a point $p$ if the minimum axis-aligned rectangle containing $g$ and $p$ is a subset of $P$. This model was introduced in 1986 by Keil~\cite{keil1986} who gave an $O(n^2)$-time algorithm for horizontally convex orthogonal polygons. $r$-guarding is NP-hard in orthogonal polygons with holes \cite{Biedl016}, but is solvable in $O(n^3)$ time in simple orthogonal polygons \cite{worman2007}. There are also linear-time approximation algorithms \cite{Lingas2012} as well
as algorithms for special cases~\cite{CulbersonReckhow1989,keil1986,PaliosTzimas2014}.

\item
\emph{Staircase-guarding ($s$-guarding).} A point guard $g$  {\em $s$-guards} any point $p$ that can be reached from $g$ by a {\em staircase}, i.e., an orthogonal path inside $P$ that is both $x$-monotone and $y$-monotone. This was introduced by Motwani et al.~\cite{MotwaniRS89} who proved that $s$-guarding is polynomial on simple orthogonal polygons. See also~\cite{MotwaniRS90}.

\item
\emph{Periscope-guarding.} A \emph{periscope guard} $g$ can see all points $p$ in which some orthogonal path inside $P$ connects $g$ to $p$ and has at most one bend. This was introduced by Gewali and Ntafos~\cite{GewaliN92}, who showed NP-hardness in 3D and gave positive results for special types of grids (i.e., sets of orthogonal line segments).

A natural generalization of periscope-guards are {\em $k$-periscope guards} in which a point guard $g$ can see all points $p$ that are connected via an orthogonal path inside $P$ with at most $k$ bends.  (In contrast to $s$-guards, monotonicity of the path is not required.) This was also studied by Gewali and Ntafos~\cite{GewaliN92}, and is an orthogonalized version of {\em $k$-visibility} guards where a guard can see along a path (not necessarily orthogonal) with up to $k$ segments \cite{Shermer91}.

Another variant is to consider length rather than number of bends. Thus, an {\em $L_1$-distance guard} $g$ (for some fixed distance-bound $D$) can see all points $p$ for which some orthogonal path from $g$ to $p$ inside $P$ has length at most $D$.\footnote{We use ``$L_1$'' to emphasize that this path {\em must} be orthogonal; the concept would make sense for non-orthogonal paths but we do not have any results for them.} We are not aware of previous results for this type of guard.

\item
\emph{Sliding cameras.} Recently there has been much interest in {\em mobile guards}, where a guard can walk along a line segment inside polygon $P$, and can see all points that it can see from some point along the line segment.  In an orthogonal setting, this type of guards becomes a {\em sliding camera}, i.e., an axis-aligned line segment $s$ inside $P$ that can see a point $p$ if the perpendicular from $p$ onto $s$ lies inside $P$. The sliding cameras model of visibility was introduced in 2011 by Katz and Morgenstern~\cite{katz2011}.  It is NP-hard in polygons with holes \cite{durocherM2013} (see also~\cite{SaeedThesis}); its complexity in simple polygons is open.
\end{itemize}

\paragraph{Related results.} We showed in an earlier paper that $r$-guarding a polygon with bounded treewidth can be solved in linear time \cite{Biedl016}. We briefly sketch here how this worked, so that we can explain why it does not transfer to $s$-guarding. The main idea was to express $r$-guarding as a restricted-distance-2 dominating set problem in a suitably defined auxiliary graph. This graph has vertices for all possible guards, all points that need watching, and as ``intermediaries'' all maximal axis-aligned rectangles inside $P$, with point $p$ adjacent to rectangle $R$ if and only if $p\in R$.  The crucial argument is that with this choice of intermediary any point of $P$ belongs to $O(f(t))$ intermediaries, where $t$ is the treewidth of the polygon and $f(\cdot)$ is a suitable function.  Therefore, one can argue that the auxiliary graph has bounded treewidth if the polygon does, and so restricted-distance-2 dominating set can be solved. A similar (and even simpler) approach works for sliding cameras \cite{BiedlCLMMV17}; here the ``intermediaries'' are maximal orthogonal line segments.

\paragraph{Our results.} This main goal of this paper is to solve the $s$-guarding problem in polygons of bounded treewidth. We first attempted an approach similar to the one used for $r$-guarding, i.e., to find suitably intermediaries and use restricted-distance-2 dominating set.  We were unsuccessful, and suspect that due to the arbitrary number of bends in staircases, no intermediaries can exist for which any point belongs to $O(f(t))$ intermediaries. Instead, we develop an entirely different approach. Note that the above guarding-models (except $r$-guarding) are defined as ``there exists an orthogonal path from $g$ to $p$ that satisfies some property''.  One can argue (see Lemma~\ref{lem:pathOnPixelation}) that
we may assume the path to run along edges of the pixelation. The guarding problem then becomes the problem of reachability in a directed graph derived from the pixelation.  This problem is polynomial in graphs of bounded treewidth,  and we hence can solve the guarding problem for $s$-guards, $k$-periscope-guards, sliding cameras, and a special case of $L_1$-distance-guards, presuming the polygon has bounded treewidth.

One crucial ingredient (similarly used in \cite{BiedlCLMMV17,Biedl016}) is that we can usually reduce the (infinite) set of possible guards to a finite set of ``candidate guards'', and the (infinite) set of points that need to be guarded to a finite set of ``watch points'' while maintaining an equivalent problem.  This is not trivial (and in fact, false for some guarding-types), and may be of independent interest since it does not require the polygon to have bounded treewidth. We discuss this in Section~\ref{sec:prelimins}.

To explain the construction for $s$-guarding, we first solve (in Section~\ref{subsec:urg}) a subproblem in which an $s$-guard can only see along a staircase in north-eastern direction. We then combine four of the obtained constructions to solve $s$-guarding (Section~\ref{subsec:sGuarding}). In Section~\ref{sec:otherTypes}, we modify the construction to solve several other orthogonal guarding variants. We conclude in Section~\ref{sec:conclusion}.

\section{Preliminaries}
\label{sec:prelimins}
Throughout the paper, let $P$ denote an orthogonal polygon (possibly with holes) with $n$ vertices. We already defined $\alpha$-guards (for $\alpha= r,s,$ periscope, etc.).
The {\em $\alpha$-guarding problem} consists of finding the minimum set of $\alpha$-guards
that can see all points in $P$.  We solve a more general problem that allows to restrict 
the set of guards and points to be guarded.
Thus, the {\em $(\Gamma,X)$-$\alpha$-guarding problem}, for 
some (possibly infinite) sets $\Gamma\subseteq P$ and $X\subseteq P$,  
consists of finding a minimum
subset $S$ of $\Gamma$ such that all points in $X$ are $\alpha$-guarded
by some point in $S$, or reporting that no such set exists.  
Note that with this, we can for example restrict guards to be only at
polygon-vertices or at the polygon-boundary, if so desired.
The standard $\alpha$-guarding problem is the same as the $(P,P)$-$\alpha$-guarding problem.

Recall that the {\em standard pixelation} of $P$ is obtained by extending a horizontal and a vertical ray inward from any reflex vertex until they hit the boundary.
This is one method of obtaining a {\em pixelation} of $P$, i.e., a partition of the polygon into axis-aligned rectangles ({\em pixels}) such that any pixel-corner is either on the
boundary of $P$ or incident to four pixels. The \emph{1-refinement} of a pixelation is the result
of partitioning every pixel into four equal-sized rectangles. See Figure~\ref{fig:pixelAndRefinement}. 

A pixelation can be seen as planar straight-line graph, with vertices at pixel-corners and edges along pixel-sides.
For ease of notation we do not distinguish between the geometric construct (pixel/pixel-corner/pixel-side) and its equivalent in the graph (face/vertex/edge).
To solve guarding problems, it usually suffices to study this graph due to the following:
\begin{lemma}
\label{lem:pathOnPixelation}
Let $P$ be a polygon with a pixelation $\Psi$. Let $\pi$ be an orthogonal path inside $P$ that connects two vertices $g,p$ of $\Psi$. Then there exists a path $\pi'$ from $g$ to $p$ 
along edges of $\Psi$ that satisfies
\begin{itemize}
\item $\pi'$ is monotone if $\pi$ was,
\item $\pi'$ has no more bends than $\pi$,
\item $\pi'$ is no longer than $\pi$.
\end{itemize}
\end{lemma}
\begin{proof}
Let $s_1,\dots,s_\ell$ be the segments of $\pi$, in order from $g$ to $p$. Let $i$ be minimal such that $s_i$ does not run along pixel-edges.
Observe that $1<i<\ell$ since $g$ and $p$ are pixel-corners and so their
incident segments are on pixel-edges. Shift $s_i$ in parallel by shortening $s_{i-1}$ until the shifted segment
resides on a pixel-edge.  (This must happen at the latest when
$s_{i-1}$ has shrunk to nothing, since the other end of $s_{i-1}$ is 
a pixel-corner by choice of $i$.)  We shorten or lengthen $s_{i+1}$ as
needed to maintain a path.  Throughout the shift, segment $s_i$ remains
within the same pixels by choice of shift, and so it remains within $P$.
One easily verifies that the path obtained after the shift satisfies the
conditions: we do not change the directions of segments, we do not add bends (and perhaps
remove some if $s_{i-1}$ shrinks to nothing), and we do not add length
(any increase of $|s_{i+1}|$ corresponds to a decrease of $|s_{i-1}|$). 
After the shift $s_i$ resides on pixel-edges, and sufficient repetition
gives the result.
\end{proof}

\subsection{Tree decompositions}
\label{sec:treewidth}
A \emph{tree decomposition} of a graph $G$ is a tree $I$ and an assignment ${\cal X}:I\rightarrow 2^{V(G)}$ of \emph{bags} to the nodes of $I$ such that (i) for any vertex $v$ of $G$, the bags containing $v$ form a connected subtree of $I$ and (ii) for any edge $(v,w)$ of $G$, some bag contains both $v$ and $w$. The width of such a decomposition is $\max_{X\in {\cal X}} |X|-1$, and the \emph{treewidth} $tw(G)$ of $G$ is the minimum width over all tree decompositions of $G$. 

We aim to prove results for polygons where the standard pixelation has bounded treewidth. Because we sometimes use the 1-refinement of $P$ instead, we need:
\begin{observation}
\label{obs:boundedTreewidthOfRefinement}
Let $P$ be a polygon with a pixelation $\Psi$ of treewidth $t$.
Then the 1-refinement of $\Psi$ has treewidth $O(t)$.
\end{observation}
\begin{proof}
Let $\mathcal{T}=(I, \mathcal{X})$ denote a tree decomposition of width $t$ of $\Psi$.
Let $X$ be a bag in $\mathcal{X}$, which contains vertices of the standard pixelation. Obtain a bag $X'$ by adding to it, for every $v\in X$,
the (up to) 8 vertices $v'$ of the 1-refinement such that some pixel of the
1-refinement contains both $v$ and $v'$. Let $\mathcal{T'}=(I, \mathcal{X'})$ be the tree decomposition consisting of
these bags $X'$, with the same adjacency structure $I$ as for $\mathcal{T}$.
Clearly, this tree decomposition has width $O(t)$, and one can easily verify
that it indeed is a tree decomposition of the 1-refinement.
\end{proof}

The standard pixelation of an $n$-vertex polygon may well have $\Omega(n^2)$ vertices in general, but not for polygons of bounded treewidth.

\begin{lemma}
Let $P$ be a polygon with $n$ vertices and treewidth $t$. 
Then the standard pixelation $\Psi$ of $P$ has $O(3^t n)$ vertices.
\end{lemma}
\begin{proof}
First observe  that there are $O(n)$ {\em boundary-corners}, i.e., pixel-corners that lie on the boundary of $P$. This holds because any boundary-corner is either a vertex of $P$ or
results when a ray from a reflex vertex hits the boundary, but there are only two such rays per reflex vertex and each hits the boundary only once.

Now consider for each boundary-corner $v$ the pixel-corners that are within (graph-theoretic) distance $t$ from $v$.  There can be only $O(3^t)$ many such pixel-corners per boundary-corner $v$, since pixel-vertices have maximum degree 4. \footnote{One may be tempted to think that there are only $O(t^2)$ such
vertices, since they form a grid-like structure.  However, it is possible for one vertex to have $3^t$ neighbours of distance $t$, e.g. if the 
graph contains a complete ternary tree of height $t$ (which can be drawn orthogonally without bends).  In all constructions that we tried,
the overall number of vertices is still $O(t^2 n)$, but proving this remains open.} Thus, there are $O(3^t n)$ pixel-corners within distance $t$ of the boundary.

So, we are done unless there exists a pixel-corner $u$ that is not within distance $t$ of the boundary.  Then for any pixel-corner
that is within distance $t$ of $u$, all four possible neighbours in the pixelation exist.  In consequence, the vertices of distance up to $t+1$ of $u$
form a diamond-shape that lies within $2t+3$ rows and columns of the pixelation; see also Figure~\ref{fig:diamond}.

The midpoints of the four diagonals of this diamond hence form the four corners of a grid with at least $t+1$ rows and columns.
Such a grid is well-known to have treewidth $t+1$, a contradiction. So, no such pixel-corner $u$ exists and $\Psi$ has $O(3^t n)$ vertices.
\end{proof}

\begin{figure}[ht]
\centering
\includegraphics[width=0.45\linewidth]{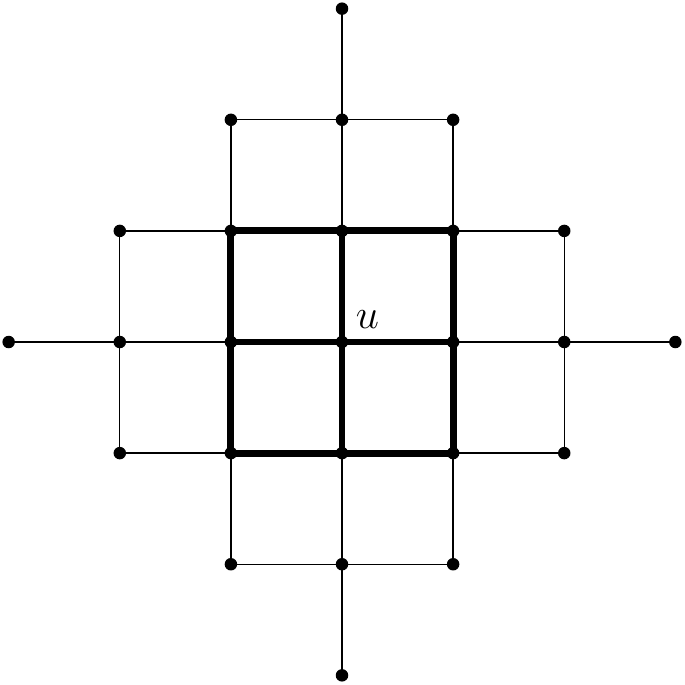}
\caption{A diamond centered at $u$.  We show $t=2$, hence the diamond-shape
lies within a $7\times 7$-grid and gives rise to a $3\times 3$-grid inside it.}
\label{fig:diamond}
\end{figure}

The 1-refinement has asymptotically the same number of vertices as the pixelation, hence it also has $O(3^t n)$ vertices.

\subsection{Reducing the problem size}
In the standard guarding problem, guards can be at an infinite number of points inside $P$, and we must guard the
infinite number of all points inside $P$.  To reduce the guarding problem to a graph problem, we must
argue that it suffices to consider a finite set of guards (we call them {\em candidate guards}) and to check that a finite set of points is guarded
(we call them {\em watch points}). Such reductions are known for $r$-guarding \cite{Biedl016} and
sliding cameras \cite{BiedlCLMMV17}.  Rather than re-proving it for each
guarding type individually, we give here a general condition under which 
such a reduction is possible.

We need a few notations.  First, all our guarding models (with the exception of sliding cameras) use {\em point guards}, i.e., guards are points that belong to $P$. Also, all guarding models are {\em symmetric}; i.e., point $g$ guards point $p$ if and only if $p$ guards $g$. We say that two guarding problems $(\Gamma,X)$ and $(\Gamma',X')$ are {\em equivalent} if given the solution of one of them, we can obtain the solution of the other one in linear time.

\begin{lemma}
\label{lem:shift}
\label{lem:simplifyGeneral}
Let $P$ be an orthogonal polygon with a pixelation $\Psi$. Consider a guarding-model $\alpha$ that uses point guards, is symmetric, and satisfies the following:
\begin{itemize}
\vspace*{-2mm}
\itemsep -2pt
\item[(a)]  For any pixel $\psi$ and any point $g\in P$, if $g$ $\alpha$-guards
	one point $p$ in the interior of $\psi$, then it $\alpha$-guards all points in $\psi$.
\item[(b)]  For any edge $e$ of a pixel and any point $g\in P$, if $g$ $\alpha$-guards
	one point $p$ in the interior of $e$, then it $\alpha$-guards all points on $e$.
\end{itemize}
Then for any (possibly infinite) sets $X,\Gamma \subseteq P$ there exist (finite) sets $X',\Gamma'$ 
such that $(\Gamma,X)$-$\alpha$-guarding and $(\Gamma',X')$-$\alpha$-guarding are equivalent.
Moreover, $X'$ and $\Gamma'$ consist of vertices of the 1-refinement of $\Psi$. 
\end{lemma}
\begin{proof}
Let $V_0$ be the vertices of the pixelation and $V_1\supseteq V_0$ 
be the vertices of the 1-refinement. We repeatedly need the following ``shifting''-operation $s(a)$
that shifts a point $a\in P$ to the nearest point in $V_1$ that act the same with respect to pixels.  Formally, for any
$a\in A$, if $a\in V_0$ then $s(a)=a$. If $a\not \in V_0$, but $a$ lies on an edge $e$ of the
standard pixelation, then let $s(a)$ be the midpoint of $e$.
If $a$ lies on neither vertex nor edge of the standard pixelation, 
then it belongs to the interior of a pixel $\psi$; let $s(a)$ be
the center of $\psi$.  We have $s(a)\in V_1$ in all cases.
For $A\subseteq P$, define $s(A):=\bigcup_{a\in A} s(a)$, and
note that $s(A)$ is finite even if $A$ is not.

Observe that some point $g$ guards a point $p$ if and only if $g$ guards
the point $s(p)$.  This is obvious if $p\in V_0$, since then $s(p)=p$.  If $p\not\in V_0$, then it is either
on an edge of a pixel (then $s(p)$ is the midpoint of that edge),
or it is in the interior of a pixel (then $s(p)$ is the center of that
pixel).  By the assumption on $\alpha$-guarding, therefore $g$ guards $s(p)$
if and only if it guards the entire edge/pixel, which it does if and only
if $g$ guards $p$.   So we can shift watch-points without affecting whether
they are guarded.  By symmetry of $\alpha$-guarding, we can also shift guards.
Namely, $g$ guards $p$ if and only if $p$ guards $g$, if and only if $p$ guards $s(g)$, if
and only if $s(g)$ guards $p$.

We claim that using $X':=s(X)$ and $\Gamma':=s(\Gamma)$ gives the result.
To see this, let $S,S'$ be solutions to the $(\Gamma,X)$-$\alpha$-guarding
and $(\Gamma',X')$-$\alpha$-guarding problem, respectively.  
Consider the set $s(S)$, which guards $X$ since $S$ does.    
Therefore $s(S)$ also guards $s(X)=X'$, and it is a solution to the 
$(\Gamma',X')$-$\alpha$-guarding problem.  So $|S'|\leq |s(S)| \leq |S|$.

For the other inequality, define a new set $S''$ as follows.
For every guard $g'\in S'\subseteq \Gamma'=s(\Gamma)$, let $g\in \Gamma$
be a guard with $s(g)=g'$ (breaking ties arbitrarily), and
add $g$ to $S''$.  Then $s(S'')=S'$,
and so $S''\subseteq \Gamma$ guards $X'$ since $S'$ does.    
By $X'=s(X)$, therefore $S''$ guards $X$.  So $S''$ is a solution to the
$(\Gamma,X)$-$\alpha$-guarding problem,  and $|S|\leq |S''| \leq |S'|$.

Putting it together, we have $|S|=|S'|$, and we can obtain one from the other via a shifting operation or its inverse. Therefore, the two problems are equivalent.
\end{proof}

It is easy to see that the conditions of Lemma~\ref{lem:shift} are satisfied for $r$-guarding, $s$-guarding and $k$-periscope guarding (for any $k$). For completeness' sake we give here the argument for $s$-guards and only for the first condition; the proof is similar for the second condition and all other guarding models.
\begin{claim}
\label{claim:interior2boundary}
In the standard pixelation, for any pixel $\psi$ and any point $g\in P$, if $g$ $s$-guards
	one point $p$ in the interior of $\psi$, then it $s$-guards all points in $\psi$.
\end{claim}
\begin{proof}
Consider the staircase $\pi$ from $g$ to $p$, and assume that it is a \texttt{NW}-staircase, all other cases are symmetric. Let $s_\texttt{E}$ and $s_\texttt{S}$ be the maximal segment within $P$ through the right and bottom side of $\psi$, respectively. Assume first that $\pi$ crosses both $s_\texttt{E}$ and $s_\texttt{S}$, and let $q$ be the first point (while walking from $g$ to $p$) that lies on one of them, say $q$ is on $s_\texttt{E}$; see Figure~\ref{fig:reduction}(a). Since $\pi$ is a $\texttt{NW}$-staircase, point $q$ must be below $s_\texttt{S}$. We can then guard any point in $\psi$ by walking from $g$ to $q$ along $\pi$, then along $s_\texttt{E}$ to the south-east corner of $\psi$, and from there with at most one bend to any point in $\psi$.

\begin{figure}[t]
\centering
\includegraphics[width=0.75\linewidth,page=2,trim=0 20 0 0,clip]{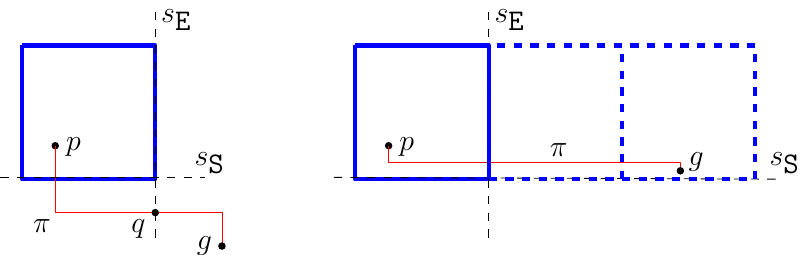}
\caption{Two cases for the proof of Claim~\ref{claim:interior2boundary}. }
\label{fig:reduction}
\end{figure}

Now, assume that $\pi$ does not cross, say, $s_\texttt{S}$ (see Figure~\ref{fig:reduction}(b)).  Since $\pi$ is a $\texttt{NW}$-staircase, therefore the entire path $\pi$ stays within the pixels directly to the right of $\psi$. These form a rectangle inside $P$, and any two points within a rectangle can $s$-guard each other.  So $g$ $s$-guards all points in $\psi$.
\end{proof}

\section{Algorithm for $(\Gamma,X)$-$s$-guarding}
\label{sec:algSGuarding}
In this section, we give a linear-time algorithm for the $(\Gamma,X)$-$s$-guarding problem on any orthogonal polygon $P$ 
with bounded treewidth. By Lemma~\ref{lem:simplifyGeneral}, we may assume that $\Gamma$ and $X$ consist of vertices of the 1-refinement of the standard pixelation. As argued earlier, the 1-refinement also has bounded treewidth. Thus, it suffices to solve the $(\Gamma, X)$-$s$-guarding where $\Gamma$ and $X$ are vertices of some pixelation $\Psi$ that has bounded treewidth.

\subsection{$(\Gamma, X)$-\urg-Guarding} 
\label{subsec:urg}
For ease of explanation, we first solve a special case where guards can look in only two of the four directions and then show how to generalize it to $s$-guarding. We say that a point \emph{$g$ \urg-guards} a point $p$ if there exists an orthogonal path $\pi$ inside $P$ from $g$ to $p$ that goes alternately \emph{north} and \emph{east}; we call $\pi$ a \emph{\urg-path}. Define \ulg-, \drg- and \dlg-guarding analogously.

\begin{wrapfigure}{r}{0.35\linewidth}
\centering
\includegraphics[width=0.60\linewidth]{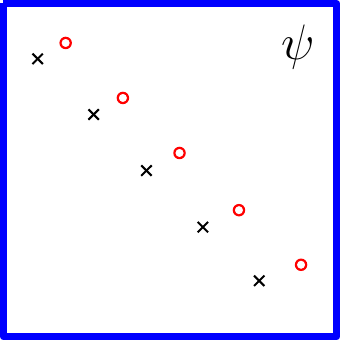}
\caption{One pixel needs many guards.}
\label{fig:notReducibleURG}
\end{wrapfigure}

Note that $\urg$-guarding does not satisfy the conditions of Lemma~\ref{lem:simplifyGeneral}; see 
e.g.~Figure~\ref{fig:notReducibleURG} where all crosses are needed to \urg-guard all circles. So, we cannot solve the $\urg$-guarding problem in general, but we can solve
$(X,\Gamma)$-\urg-guarding  since we already know that $X$ and $\Gamma$ are vertices of the pixelation.

\paragraph{Constructing an auxiliary graph $H$.} Define graph $H$ to be the graph of the pixelation of $P$ and direct each edge of $H$ toward north or east;
see Figure~\ref{fig:graphHRefine} for an example. By assumption $X\subseteq V(H)$ and $\Gamma\subseteq V(H)$.

\begin{figure}[t]
\centering
\includegraphics[width=0.60\linewidth,trim=250 0 0 0,clip]{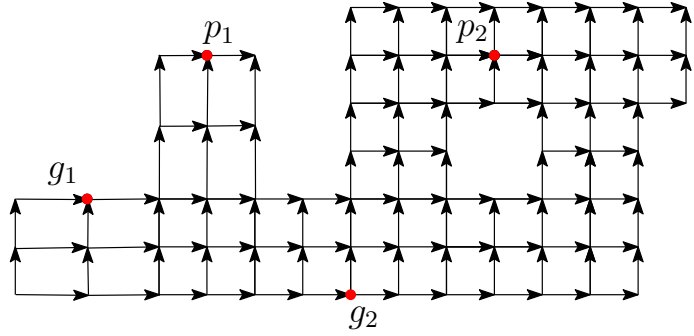}
\caption{The graph $H$ corresponding to \urg-guarding the polygon of Figure~\ref{fig:pixelAndRefinement}.  Guards and points have been shifted to pixelation-vertices.}
\label{fig:graphHRefine}
\end{figure}

By Lemma~\ref{lem:pathOnPixelation}, there exists an \urg-path from guard $g\in \Gamma$ to point $p\in X$ if and only if there exists one along pixelation-edges.  With our choice of edge-directions for $H$, hence there exists such a \urg-path if and only if there exists a directed path from $g$ to $p$ in $H$. Thus, $(\Gamma,X)$-\urg-guarding reduces to the following problem which we call {\em reachability-cover}: given a directed graph $G$ and vertex sets $A$ and $B$, find a minimum set $S\subseteq A$ such that for any $t\in B$ there exists an $s\in S$ with a directed path from $s$ to $t$. $(X,\Gamma)$-\urg-guarding  is equivalent to reachability-cover in graph $H$ using $A:=\Gamma$ and $B:=X$.

Reachability-cover is NP-hard because set cover can easily be expressed in it. We now argue that reachability-cover can be solved in graphs of bounded treewidth, by appealing to monadic second order logic or MSOL (see \cite{Courcelle97} for an overview).  Briefly, this means that the desired graph property can be expressed as
a logical formula that may have quantifications, but only on variables and sets. Courcelle's theorem states that any problem expressible in MSOL can be solved in linear time on 
graphs of bounded treewidth~\cite{Courcelle90}. (Courcelle's original result was only for decision problems, but it can easily be generalized to minimization problems.)
Define \texttt{Reachability}$(u,v,G)$ to be the property that there exists a directed
path from $u$ to $v$ in a directed graph $G$.  This can be expressed in MSOL \cite{Courcelle97}.
Consequently, the $(\Gamma, X)$-\urg-guarding problem can be expressed in MSOL as follows:
\[
\exists S \subseteq \Gamma: \forall p\in X: \exists g\in S: \texttt{Reachability}(g,p,H).
\]
So, we can solve the $(\Gamma, X)$-\urg-guarding problem if $\Gamma$ and $X$ are vertices of a given pixelation that has bounded treewidth.

\subsection{$(\Gamma, X)$-$s$-guarding}
\label{subsec:sGuarding}
Solving the $(\Gamma,X)$-$s$-guarding problem now becomes very simple, by exploiting that a guard $g$ $s$-guards a point $p$ if only if $g$ $\beta$-guards $p$ for some 
$\beta\in$\{\urg, \ulg, \drg, \dlg\}. We can solve the $(\Gamma,X)$-$\beta$-guarding problem 
for $\beta\neq \urg$ similarly as in the previous section, by directing the auxiliary graph $H$ according to the directions
we wish to take.  Let $H_\urg,H_\ulg,H_\drg,H_\dlg$ be the four copies of graph $H$
(directed in four different ways) that we get. Define a new auxiliary graph $H^*$ as follows
(see also Figure~\ref{fig:fourGraphHExample}): initially, let $H^*:=H_\urg\cup H_\ulg\cup H_\drg\cup H_\dlg$. 
For each $g\in\Gamma$, add to $H^*$ a new vertex $v^\Gamma(g)$ and the directed edges $(v^\Gamma(g),v_\beta(g))$ 
where $v_{\beta}(g)$ (for $\beta\in$\{\urg, \ulg, \drg, \dlg\}) is the vertex in $H_{\beta})$ corresponding to $g$. Similarly, for each $p\in X$, add to $H^*$ a new vertex $v^X(p)$ and the directed edges $(v_\beta(p),v^X(p))$ for $\beta\in$\{\urg, \ulg, \drg, \dlg\}).

If some guard $g$ $s$-guards a point $p$, then there exists a $\beta$-path from $g$ to $p$ inside $P$ for some $\beta\in$\{\urg, \ulg, \drg, \dlg\}.  We can turn this path into a $\beta$-path along pixelation-edges by Lemma~\ref{lem:pathOnPixelation}, and therefore
find a path from $v^\Gamma(g)$ to $v^X(p)$ by going to $H_\beta$ and following the path within it.  Vice versa, any directed path from $v^\Gamma(g)$ to $v^X(p)$
must stay inside $H_\beta$ for some $\beta\in$\{\urg, \ulg, \drg, \dlg\} since $v^\Gamma(g)$ is a source and
$v^X(p)$ is a sink.  Therefore, $(\Gamma,X)$-$s$-guarding is the same as reachability-cover in $H^*$
with respect to the sets $V(\Gamma):=\{v^\Gamma(g): g \in \Gamma\}$ and $V(X):=\{v^X(p): p\in X\}$.

\begin{figure}[t]
\centering
\includegraphics[width=0.60\textwidth,trim=0 100 0 0,clip]{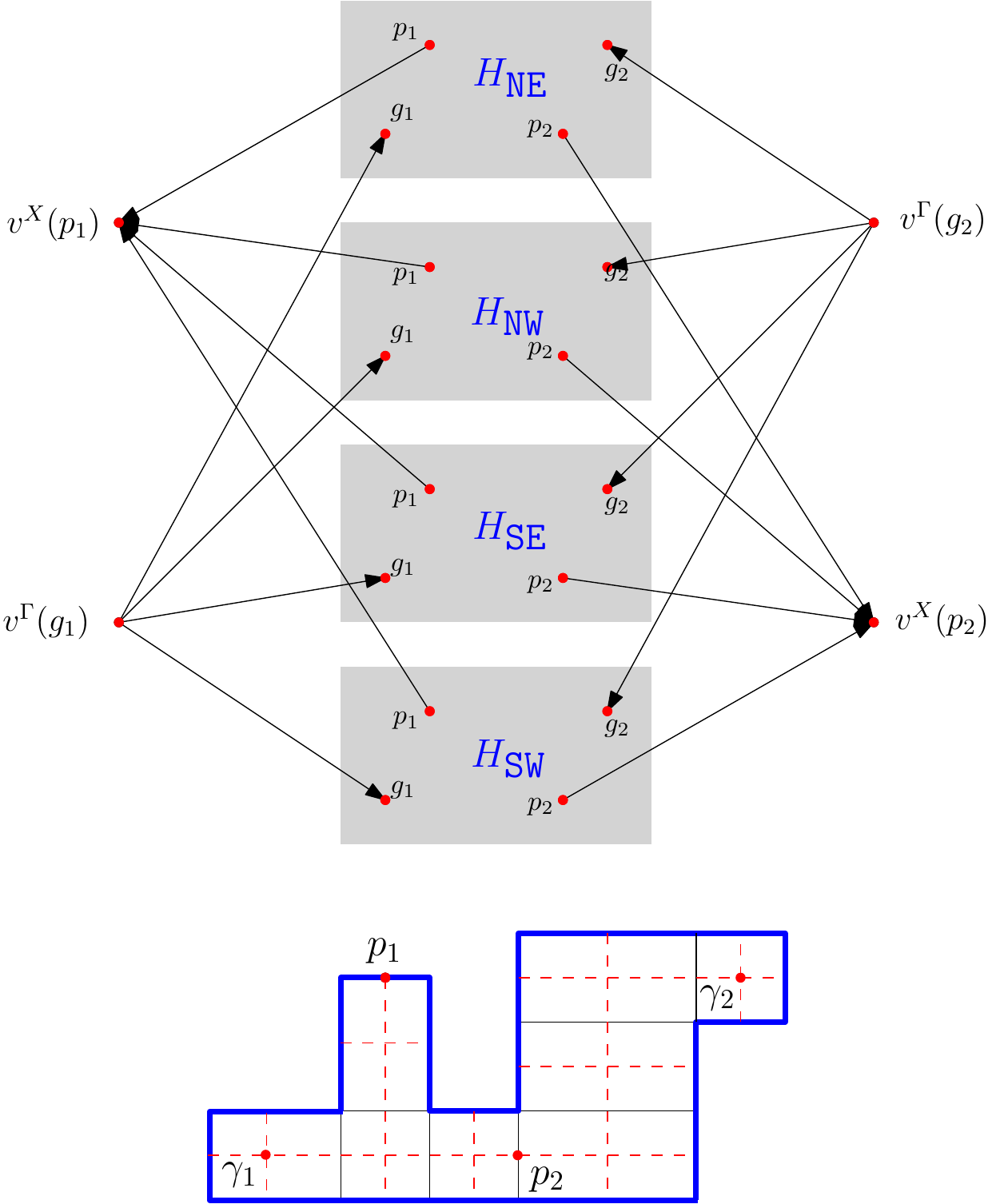}
\caption{The construction of graph $H^*$.}
\label{fig:fourGraphHExample}
\end{figure}

It remains to argue that $H^*$ has bounded treewidth.  To do so, take a tree decomposition of ${\cal T}=(I,{\cal X})$ of the pixelation.  For each bag $X\in {\cal X}$, create a new
bag $X'$ by replacing $x\in X$ with the four copies $v_{\beta}(x)$, as well as $v^\Gamma(x)$ if $x\in \Gamma$ and $v^X(x)$ if $x\in P$.  Clearly ${\cal T}':=(I,\{X'\})$ is tree decomposition of $H^*$ whose width is proportional to the one of ${\cal T}$.

Now, we put it all together.  Assume $P$ has bounded treewidth, hence its standard pixelation has bounded treewidth and $O(n)$ edges, and so does its 1-refinement.  This is the pixelation we use to obtain $H^*$, therefore $H^*$ also has bounded treewidth and $O(n)$ edges. We can apply Courcelle's theorem to solve reachability-cover in $H^*$ and obtain:
\begin{theorem}
\label{thm:mainSGuardingResult}
Let $P$ be an orthogonal polygon with bounded treewidth. Then, there exists an linear-time algorithm for the $(\Gamma, X)$-$s$-guarding problem on $P$.
\end{theorem}

Our result as stated is of mostly theoretical interest, since the dependence of the run-time on the treewidth $t$ is not clear (and in general, may be quite high if applying Courcelle's theorem).  However, it is not hard to solve the reachability-cover problem directly, via dynamic programming in a tree decomposition, in time $O(2^t n)$ for a graph with treewidth $t$.  (We leave the details to the reader; see for example \cite{Bod08} for other examples of such dynamic programming algorithms.)  As such, the algorithm should be quite feasible to implement for small values of $t$.

\section{Other Guarding Types}
\label{sec:otherTypes}
In this section, we show how similar methods apply to other types of orthogonal guarding. The main difference is that we need edge-weights on the auxiliary graph.  To solve the guarding problem, we hence use a version of reachability-cover defined as follows. The {\em $(G,A,B,W)$-bounded-reachability-cover problem} has as input an edge-weighted directed graph $G$, two vertex sets $A$ and $B$, and a length-bound $W$. The objective is to find a minimum-cardinality set $S\subseteq A$ such that for any $t\in B$ there exists an $s\in S$ with a directed path from $s$ to $t$ that has length at most $W$. We need to argue that this problem is solvable if $G$ has bounded treewidth, at least if $W$ is sufficiently small. Recall that reachability-cover can be expressed in monadic second-order logic. Arnborg et al.~\cite{ArnborgLS91} introduced the class of \emph{extended} monadic second-order problems which allow integer weights on the input. They showed that problems expressible in extended monadic second-order logic can be solved on graphs of bounded treewidth, with a run-time that is polynomial in the graph-size and the maximum weight.

\subsection{$L_1$-distance guarding}
We first study the $L_1$-distance guarding problem. Recall that this means that a guard $g$ can see (for a given upper bound $D$) all points $p$ for which there exists an orthogonal path $\pi$ of length at most $D$. Our approach works only under some restrictions, but is so simple that we will describe it anyway.

We have not been able to solve the $L_1$-distance guarding problem for all polygons of bounded treewidth, for two reasons.  First, the $L_1$-distance does not satisfy the conditions of Lemma~\ref{lem:simplifyGeneral}.  (For example, if there is a candidate guard at the top-right corner of the pixel in Figure~\ref{fig:notReducibleURG}, then for a suitable value of $D$ it can see all the  circle-points but none of the cross-points.)  Thus we only consider the $(\Gamma,X)$-$L_1$-distance guarding problem where $\Gamma$ and $X$ are vertices of a pixelation that has bounded treewidth.  The second problem is that the bounded-reachability-cover problem is solved in run-time that depends on the maximum weight. For this to be polynomial, we must assume that all edges of the input-polygon have integer length that is polynomial in $n$.

Let $\Gamma$ and $X$ be subsets of the vertices of some pixelation $\Psi$ of $P$.  Let $H_{\text{dist}}$ be the auxiliary graph obtained from the pixelation graph by making all edges bi-directional.  Set the weight of each edge to be its length.  If a guard $g\in \Gamma$ sees a point $p\in X$ in the $L_1$-distance guarding model (with distance-bound $D$), then there exists a path $\pi$ from $g$ to $p$ that has length at most $D$.  By Lemma~\ref{lem:pathOnPixelation}, we may assume that $\pi$ runs along pixel-edges.  Hence $\pi$ gives rise to a directed path in $H_{\text{dist}}$ of length less than $D$.  Vice versa, any such path in $H_{\text{dist}}$ means that $g$ can $L_1$-distance-guard $p$.  In consequence, the $(\Gamma,X)$-$L_1$-distance guarding problem is the same as the $(\Gamma,X,H_{\text{dist}},D)$-bounded-reachability-cover problem, and we have:
\begin{theorem}
The $(\Gamma,X)$-$L_1$-distance guarding problem can be solved in a polygon $P$ with polynomial integral edge-lengths, presuming that $P$ has a pixelation $\Psi$ of bounded treewidth that has $\Gamma$ and $X$ at its vertices.
\end{theorem}

\subsection{$k$-periscope guarding}
Next we turn to $k$-periscope guards, where guard $g$ can see a point $p$ if there exists an orthogonal path $\pi$ from $g$ to $p$ with at most $k$ bends.  Similarly as for
$s$-guards, one can show that $k$-periscope guards satisfy the conditions of Lemma~\ref{lem:simplifyGeneral}, so we can simplify any input to use a discrete set of candidate guards and watch points that are vertices of the 1-refinement of the standard pixelation.  Any path $\pi$ from such a candidate guard to such a watch point can be assumed to reside on pixel-edges by Lemma~\ref{lem:pathOnPixelation}.

\begin{wrapfigure}{r}{0.35\linewidth}
\centering
\includegraphics[width=0.55\linewidth]{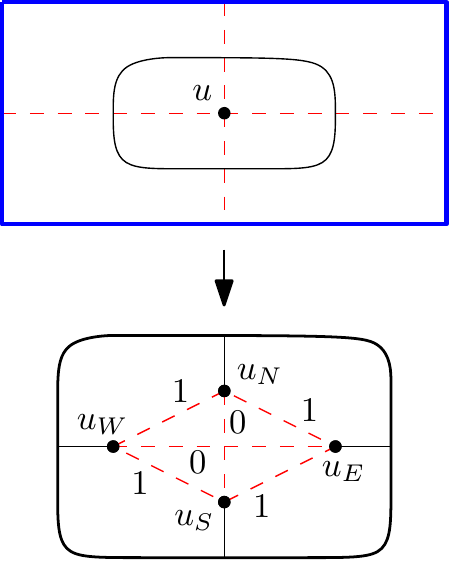}
\caption{Adding $K_4$.}
\label{fig:ak4}
\end{wrapfigure}

For $k$-periscope guarding, we define an auxiliary graph $H_{\text{peri}}$ based on the graph of the pixelation, but modify it near each vertex and add weights to encode the number of bends, rather than the length, of a path. If $u$ is a vertex of the pixelation, then replace it with a $K_4$ as shown in Figure~\ref{fig:ak4}. We denote this copy of $K_4$ by $K_4^u$, and let its four vertices be $u_N,u_S,u_W$ and $u_E$ according to compass directions.  For a vertex $u$ on the boundary of $P$ we omit those vertices in $K_4^u$ that would
fall outside $P$.  We connect copies $K_4^u$ and $K_4^v$ of a pixel-edge $(u,v)$ in the natural way, e.g. if $(u,v)$ was vertical with $u$ below $v$, then we connect $u_N$ to $v_S$.  All edges are bidirectional.

For any $g\in \Gamma$, define a new vertex $v^\Gamma(g)$ and add edges from it to all of $g_N,g_S,g_E,g_W$ that exist in the graph. For any $p\in X$, define a new vertex $v^X(p)$ and add edges from all of $p_N,p_S,p_E,p_W$ to $v^X(p)$.  Set all edge weights to 0, except for the ``diagonal'' edges between consecutive vertices of a $K_4$, which have weight 1 as shown Figure~\ref{fig:ak4}.

Clearly, $g\in \Gamma$ can see $p\in X$ (in the $k$-periscope guarding model) if and only if there is a directed path from $v^\Gamma(g)$ to $v^X(p)$ in the constructed graph that uses at most $k$ diagonal edges, i.e., that has length at most $k$.  Thus the $k$-periscope guarding model reduces to bounded-reachability-cover.  

We may assume that $k\leq n$, since any guard can see the entire polygon in the $n$-periscope guard model.  Therefore the  parameters in the bounded-reachability-cover problem are polynomial in $n$. We can hence solve the $k$-periscope guarding problem in polynomial time in any polygon of bounded treewidth. Note that the run-time depends polynomially on $k$, so $k$ need not be a constant. We  conclude:
\begin{theorem}
$k$-periscope guarding is polynomial in polygons of bounded treewidth.
\end{theorem}

\subsection{Sliding cameras}
It was already known that the sliding camera problem is polynomial in polygons of bounded treewidth~\cite{BiedlCLMMV17}.  However, using much the same auxiliary graph as in the previous subsection we can get a second (and in our opinion, easier) method of obtaining this result.

We solve the $(\Gamma,X)$-sliding camera guarding problem, for some set of sliding cameras $\Gamma$ (which are segments inside $P$) and watch points $X$. It was argued in \cite{BiedlCLMMV17} that we may assume $\Gamma$ to be a finite set of maximal segments that lie along the standard pixelation; in particular the endpoints of candidate guards are pixel-vertices. As for $X$, we cannot apply Lemma~\ref{lem:simplifyGeneral} directly, since sliding cameras are not point guards and hence not symmetric. But sliding cameras do satisfy conditions (a) and (b) of Lemma~\ref{lem:simplifyGeneral}. As one can easily verify by following the proof, we may therefore assume $X$ to consist of pixel-vertices of the 1-refinement. (A similar result was also argued in \cite{BiedlCLMMV17}.)  

We build an auxiliary graph $H_{\text{slide}}$ almost exactly as in the previous subsection. Thus, start with the graph of the 1-refinement of the standard pixelation. Replace every vertex by a $K_4$, weighted as before.  (All other edges receive weight 0.) For each $p\in X$, define a new vertex $v^X(p)$ and connect it as in the previous subsection, i.e., add edges from $p_N,p_E,p_W,p_S$ to $v^X(p)$. For any sliding camera $\gamma\in \Gamma$, add a new vertex $v^\Gamma(\gamma)$. The only new thing is how these vertices get connected. If $\gamma$ is horizontal, then add an edge from $v^\Gamma(\gamma)$ to $g_W$, where $g_W$ is the left endpoint of $\gamma$. If $\gamma$ is vertical, then add an edge from $v^\Gamma(\gamma)$ to $g_N$, where $g_N$ is the top endpoint of $\gamma$.
\begin{claim}
A sliding camera $\gamma\in \Gamma$ can see a point $p$ if and only if there
exists a directed path from $v^\Gamma(\gamma)$ to $v^X(p)$ with length at most 1.
\end{claim}
\begin{proof}
We assume that $\gamma$ is horizontal, the other case is symmetric.  If $\gamma$ can see $p$, then there exists some point $c\in \gamma$ such that the vertical line segment $\overline{cp}$ is inside $P$.  Let $g_W$ be the left endpoint of $\gamma$.  Then the orthogonal path $\pi$ from $g_W$ to $c$ to $p$ lies inside $P$.  Using the corresponding path from $g_W$ to $c$ to $p_S$ or $p_N$ in $H$, we obtain the desired directed path which has length 1 if $c\neq p$ and length 0 otherwise.  

Vice versa, if there exists a directed path from $v^\Gamma(\gamma)$ to $v^X(p)$ with length at most 1 then there exists a path $\pi$ from $g_W$ to $p_S$ or $p_N$ that uses 
at most one diagonal edge.  The part of $\pi$ before this diagonal edge corresponds to horizontal edges along the line through $g_W$; all such horizontal edges belong to $\gamma$ since $\gamma$ is a maximal segment. Thus the diagonal edge belongs to a point $c\in \gamma$.  The part of $\pi$ after this diagonal edge corresponds to vertical edges along the
line through $p$, and hence (after possible shortening of the path) goes along the line segment $\overline{cp}$, which hence must be inside $P$. So, $\gamma$ can see $p$.
\end{proof}

Therefore the sliding camera problem reduces to a bounded-reachability-cover problem where all weights are at most 1; this can be solved
in polynomial (in fact, linear) time if the polygon has bounded treewidth. We conclude:
\begin{theorem}[see also \cite{BiedlCLMMV17}]
Sliding camera guarding is polynomial in polygons of bounded treewidth.
\end{theorem}

\section{Conclusion}
\label{sec:conclusion}
In this paper, we gave algorithms for guarding orthogonal polygons of bounded treewidth. We considered various models of orthogonal guarding, and solved the guarding problem on such polygons for $s$-guards, $k$-periscope guards, and sliding cameras, and some other related guarding types.

As for open problems, the main question is whether these results could be used to obtain better approximation algorithms. Baker's method~\cite{Baker94} yields a PTAS for many problems in planar graphs by splitting the graph into graphs of bounded treewidth and combining solutions suitably.  However, this requires the problems to be ``local'' in some sense, and the guarding problems considered here are not local in that a guard may see a point whose distance in the graph of the pixelation is very far, which seems to make Baker's approach infeasible. Are these guarding problems APX-hard in polygons with holes?

\bibliographystyle{plain}
\bibliography{ref}

\end{document}